\documentclass[amsmath, 12pt]{article}
\usepackage{amsfonts}
\usepackage{amscd,amssymb,amsthm,amsmath}
\usepackage{euscript}
\usepackage{mathrsfs}
\usepackage{euscript}
\usepackage{esvect}
\usepackage{comment}
\usepackage{xcolor}
\usepackage{hyperref}
\usepackage{lmodern}
\usepackage{bm}
\usepackage{graphicx}

\newtheorem{theorem}{Theorem}

\begin{document}

\title{New Three and Four-Dimensional Toric and Burst-Error-Correcting Quantum Codes}

\author{Cibele Cristina Trinca\thanks{The author is with the Department of Biotechnology and Bioprocess Engineering, Federal University of Tocantins, Gurupi-TO, Brazil (e-mail: cibtrinca@yahoo.com.br).}, Reginaldo Palazzo Jr.\thanks{The author is with the School of Electrical and Computer Engineering, State University of Campinas, Brazil (e-mail: palazzo@dt.fee.unicamp.br).}, Ricardo Augusto Watanabe\thanks{The author is with the School of Mathematics, Statistics and Scientific Computing, State University of Campinas, Brazil (e-mail: ricardoaw18@gmail.com).}, \\ Clarice Dias de Albuquerque\thanks{The author is with the Science and Technology Center, Federal University of Cariri, Juazeiro do Norte, Brazil (e-mail: clarice.albuquerque@ufca.edu.br).}, J. Carmelo Interlando\thanks{The author is with the Department of Mathematics and Statistics, San Diego State University, San Diego, CA, USA (e-mail: carmelo.interlando@sdsu.edu).} and \\ Antonio Aparecido de Andrade\thanks{The author is with the Department of Mathematics, S\~{a}o Paulo State University, Brazil (e-mail: antonio.andrade@unesp.br).}}

\date{\today}
\maketitle

\begin{abstract}
\noindent Ongoing research and experiments have enabled quantum memory to realize the storage of qubits. On the other hand, interleaving techniques are used to deal with burst of errors. Effective interleaving techniques for combating burst of errors by using classical error-correcting codes have been proposed in several articles found in the literature, however, to the best of our knowledge, little is known regarding interleaving techniques for combating clusters of errors in topological quantum error-correcting codes. Motivated by that, in this work, we present new three and four-dimensional toric quantum codes which are featured by lattice codes and apply a quantum interleaving method to such new three and four-dimensional toric quantum codes. By applying such a method to these new codes we provide new three and four-dimensional quantum burst-error-correcting codes. As a consequence, new three and four-dimensional toric and burst-error-correcting quantum codes are obtained which have better information rates than those three and four-dimensional toric quantum codes from the literature. In addition to these proposed three and four-dimensional quantum burst-error-correcting codes improve such information rates, they can be used for burst-error-correction in errors which are located, quantum data stored and quantum channels with memory. 
\end{abstract}



\section{Introduction} 

Quantum system interactions with the environment around can destroy the superposition states causing loss of information. This process is called \textit{decoherence}. This problem can be overcome by isolating the system from its surroundings. However, for larger systems, this task is quite complicated. One way to overcome such a  difficulty is to use quantum error-correcting codes. These codes work by encoding the quantum states in order to make them resistant to the action of noise, and then decoding them at the moment in which the states are to be recovered. 

Constructions of quantum error-correcting codes are strongly based on their classical counterparts; however, there are some fundamental differences between classical and quantum information, e.g., (i) copying a qubit is physically impossible \cite{Kitaev1} and (ii) measurements destroy the quantum state in most cases, which in turn prevent its recovery \cite{Kitaev1}.  

Within the class of stabilizer quantum error-correcting codes, Kitaev introduced an alternative language using topology \cite{Kitaev1}. His proposal was to use certain properties of particles confined to a plane to perform topological quantum computation. That terminology comes from the fact that those properties are related to the topology of the physical system. Continuous deformations caused by the environment can alter those properties and, as consequence, we would naturally have quantum computation resistant to errors.

Kitaev started that investigation via toric codes \cite{Kitaev}. To construct them, qubits are associated to the edges of a square lattice on the torus; hence, the total number of edges is equal to the length of the code. The stabilizer operators are related to the vertices and the faces of the square lattice, and the coded qubits are determined according to the genus of the surface. Lastly, the distance of the code is determined by the homology group of the surface. Toric codes can be generalized to a class known as topological quantum codes by considering other two-dimensional orientable surfaces different from the torus.

In the quantum regime, quantum errors can be independent or correlated in space and time. Hence there are counterparts of quantum random error-correcting codes \cite{Wilde} and quantum burst error-correcting codes \cite{Kawabata}. Analogously to the classical case, quantum channels usually have memory \cite{Werner} or introduce errors which are located \cite{Caruso}, that is, quantum burst errors. 

The construction and investigation of quantum burst error-correcting codes have received far less attention compared to the development of standard quantum error-correcting codes or entanglement-assisted quantum error-correcting codes \cite{Wilde}. 

In quantum computing, quantum memory is the quantum-mechanical version of ordinary computer memory. Whereas ordinary memory stores information as binary states (represented by ``1"s and ``0"s), quantum memory stores a quantum state for later retrieval. These states hold useful computational information known as qubits. Unlike the classical memory of everyday computers, the states stored in quantum memory can be in a quantum superposition, giving much more practical flexibility in quantum algorithms than classical information storage.

Quantum memory is essential for the development of many devices in quantum information processing, including a synchronization tool that can match the various processes in a quantum computer, a quantum gate that maintains the identity of any state, and a mechanism for converting predetermined photons into on-demand photons. Quantum memory can be used in many aspects, such as quantum computing and quantum communication. Continuous research and experiments have enabled quantum memory to realize the storage of qubits \cite{Lvovsky}. 

Our contribution in this work is to present new three and four-dimensional toric quantum codes from lattice codes and apply a quantum interleaving method to such three and four-dimensional toric quantum codes. By applying such a method to these new codes we provide new three and four-dimensional quantum burst-error-correcting codes. As a consequence, new three and four-dimensional toric and burst-error-correcting quantum codes are obtained which have better information rates than those three and four-dimensional toric quantum codes from the literature.

This paper is organized as it follows. Sections II and III review previous results of lattice theory and toric quantum codes, respectively. Sections III and IV provide new three and four-dimensional toric quantum codes from lattice codes, respectively. In Section V a quantum interleaving method is applied to such new three and four-dimensional toric quantum codes to obtain new three and four quantum burst-error-correcting codes.

\section{Lattice}\label{lattice}
The background material presented in this section can be found in \cite{Conway,ijam1,ijam2,ijam3,ClariceArtigo,ClariceQIP,CibeleQIP}. 

Let $\{ \bm a_1, \ldots, \bm a_n\}$ be a basis for the $n$-dimensional real Euclidean space, $\mathbb R^n$, where $n$ is a positive integer. An $n$-dimensional lattice $\Lambda$ is the set of all points of the form $u_1 \bm a_1 + \cdots + u_n \bm a_n$ with $u_1,\ldots, u_n$ being integers. Thus, $\Lambda$ is a discrete additive subgroup of $\mathbb R^n$. This property leads to the study of subgroups (sublattices) and coset decompositions (partitions). An algebraic way to obtain sublattices from lattices is via a scaling matrix $A$ with integer entries. Given a lattice $\Lambda$, a sublattice $\Lambda ^{\prime}=A \Lambda$ can be obtained by transforming each vector $\lambda \in \Lambda$ to $\lambda ^{\prime} \in \Lambda ^{\prime}$ according to $\lambda ^{\prime} = A \lambda$.

Every building block that fills the entire space with one lattice point in each region is called a \textit{fundamental region} of the lattice $\Lambda$. There are several ways to choose a fundamental region for a lattice $\Lambda$, however the volume of the fundamental region is uniquely determined by $\Lambda$ \cite{ijam1,ClariceArtigo,forney}. Let $V(\Lambda)$ denote the volume of a fundamental region of the $n$-dimensional lattice $\Lambda$. For a sublattice $\Lambda ^{\prime} = A \Lambda$, we have that $\dfrac{V(A \Lambda)}{V(\Lambda)} = |\det A|$ and the set of the cosets of $\Lambda'$ in $\Lambda$ defines a \textit{lattice code} \cite{forney}.

\section{Toric Quantum Codes}\label{ToricCodes}

The material presented in this section can be found in \cite{Kitaev1, Kitaev, ClariceArtigo,CibeleQIP}. A quantum error-correcting code (QEC) is the image of a linear mapping from the $2^{k}$-dimensional Hilbert space $H^{k}$ to the $2^{n}$-dimensional Hilbert space $H^{n}$, where $k<n$. The codewords are the vectors in the $2^{n}$-dimensional space. The \textit{minimum distance d} of a quantum error-correcting code $C$ is the minimum distance between any two distinct codewords, that is, the minimum Hamming weight of a nonzero codeword. A quantum error-correcting code $C$ of length $n$, dimension $k$ and minimum distance $d$ is denoted by $[[n,k,d]]$. A code with minimum distance $d$ is able to correct up to $t$ errors, where $t=\lfloor \frac{d-1}{2} \rfloor$ \cite{Shor}. 

A stabilizer code $C$ is the simultaneous eigenspace with eigenvalue 1 comprising all the elements of an Abelian subgroup $S$ of the Pauli group $P_{n}$, called the \textit{stabilizer group}. The elements of the Pauli group on $n$ qubits are given by 
\begin{align*}
P_{n}=\{\pm I, \pm iI, \pm X, \pm iX, \pm Y, \pm iY, \pm Z, \pm iZ\}^{\otimes n}, \ \mathrm{where}
\end{align*}
\begin{align*}
I=\left( \begin{array}{cc}
                    1 & 0 \\
                    0 & 1 \\
                    \end{array}
                    \right), \ X=\sigma_{x}=\left( \begin{array}{cc}
                    0 & 1 \\
                    1 & 0 \\
                    \end{array}
                    \right),
\end{align*}
\begin{equation}
Y=\sigma_{y}=\left( \begin{array}{cc}
                    0 & -i \\
                    i & 0 \\
                    \end{array}
                    \right) \ \mathrm{and} \ Z=\sigma_{z}=\left( \begin{array}{cc}
                    1 & 0 \\
                    0 & -1 \\
                    \end{array}
                    \right).
\end{equation}

Thus, $C=\{ \mid\psi \rangle \in H^{n} \ \mid \ M\mid \psi \rangle = \ \mid\psi \rangle, \ \forall \ M\in S \}$ \cite{Gott}.

Kitaev's toric codes form a subclass of stabilizer codes and they are defined in a $q\times q$ square lattice of the torus (Figure 1). Qubits are in one-to-one correspondence with the edges of the square lattice. The parameters of this class of codes are $[[2q^{2},2,q]]$, where the code length $n$ equals the number of edges $|E|=2q^{2}$ of the square lattice. The number of encoded qubits is dependent on the genus of the orientable surface. In particular, codes constructed from orientable surfaces $gT$ (connected sum of $g$ tori $T$) encode $k=2g$ qubits. Thus, codes constructed from the torus, an orientable surface of genus 1, have $k=2$ encoded qubits. The distance is the minimum between the number of edges contained in the smallest homologically nontrivial cycle of the lattice and the number of edges contained in the smallest homologically nontrivial cycle of the dual lattice. Recall that the square lattice is self-dual and a homologically nontrivial cycle is a path of edges in the lattice which cannot be contracted to a face. Therefore the smallest of these two paths corresponds to the orthogonal axes either of the lattice or of the dual lattice. Consequently, $d=q$ \cite{DennisKitaev}.

\begin{figure}[ht]%
\centering
\includegraphics[width=0.5\textwidth]{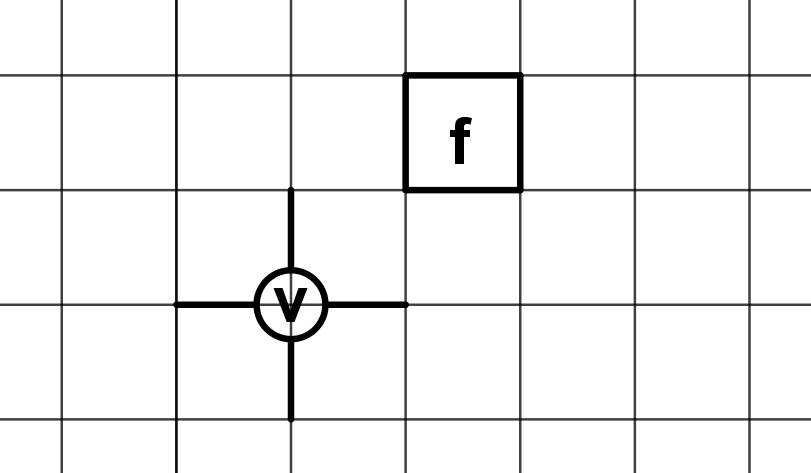}
\caption{Square lattice of the torus, from \cite{ClariceArtigo}.}
\label{fig:1}
\end{figure}

The stabilizer operators are associated with each vertex and each face of the square lattice (lattice) (Figure \ref{fig:1}). Given a vertex $v\in V$, the vertex operator $A_{v}$ is defined by the tensor product of $\sigma_{x}$ -- corresponding to each one of the four edges which have $v$ as a common vertex -- and the operator identity acting on the remaining qubits. Analogously, given a face $f\in F$, the face operator $B_{f}$ is defined by the tensor product $\sigma_{z}$ -- corresponding to each one of the four edges forming the boundary of the face $f$ -- and the operator identity acting on the remaining qubits. In particular, from \cite{ClariceArtigo},  
\begin{equation}    
A_{v}=\bigotimes_{j\in E} \sigma_{x}^{\delta (j\in E_{v})} \ \mathrm{and} \ B_{f}=\bigotimes_{j\in E} \sigma_{z}^{\delta (j\in E_{f})},
\end{equation}
\noindent where $\delta$ is the Kronecker delta. 

The toric code consists of the space fixed by the operators $A_{v}$ and $B_{f}$ and it is given as 
\begin{equation}
C=\{\mid\psi \rangle \in H^{n} \ \mid \ A_{v}\mid\psi \rangle = \ \mid\psi \rangle \ \mathrm{and} \ B_{f} \mid\psi \rangle = \ \mid\psi \rangle , \ \forall \ v,f \}.
\end{equation}

The dimension of $C$ is 4, that is, $C$ encodes $k=2$ qubits.

\section{New Three-Dimensional Toric Quantum Code from a Lattice Code}\label{three}

Three-dimensional toric quantum codes are studied in \cite{Temperature}. Under the algebraic point of view, such toric quantum codes can be characterized as the group consisting of the cosets of $q \mathbb{Z}^{3}$ in $\mathbb{Z}^{3}$ \cite{Edson}, which in turn is isomorphic to $\mathbb{Z}_{q} \times \mathbb{Z}_{q} \times \mathbb{Z}_{q} = \mathbb{Z}_{q}^{3}$, where $q$ is a positive integer. The identifications of the opposite faces of the region delimited by $\mathbb{Z}_{q} \times \mathbb{Z}_{q} \times \mathbb{Z}_{q}$ result in its identification with the three-dimensional torus denoted by $T^{3}$ which has genus $g=1$; for the sake of simplicity, we call this region \textit{lattice} or $q\times q \times q = q^{3}$ \textit{cubic lattice}. The volume associated with the lattice $\mathbb{Z}_{q} \times \mathbb{Z}_{q} \times \mathbb{Z}_{q}$ is $q^{3}$. 

Since each face belongs simultaneously to two cubes of the $q^{3}$ cubic lattice (lattice), there are
$3q^{3}$ faces, that is, $n=3q^{3}$ qubits, where $n$ is the length of the code that is given by the number of faces of the $q^{3}$ cubic lattice. In the construction of these three-dimensional toric quantum codes \cite{Temperature}, such qubits are in a biunivocal correspondence with the faces of the $q^{3}$ cubes of the $q^{3}$ cubic lattice. Since $\mathbb{Z}^{3}$ is a self-dual lattice, then the code distance which is defined as being the minimum number of faces in the $q^{3}$ cubic lattice between two codewords is equal to $q$ \cite{Edson,Temperature}. Consequently, this class of codes has parameters $[[3q^{3},3,q]]$. 

For the rest of this section, let $q=7$. The goal of this section is to construct a three-dimensional toric quantum code by using the fundamental region of a sublattice of the lattice $\mathbb{Z}^{3}$ that has volume 7, since 7 divides the volume $q^{3}=7^{3}$ of the lattice $\mathbb{Z}_{7} \times \mathbb{Z}_{7} \times \mathbb{Z}_{7}$. 

The authors in \cite{perfcodes} construct a classical single-error-correcting code in dimension three in the $7^{3}$ cubic lattice which has 49 codewords. Each codeword of this classical code has the Lee sphere of radius 1 in 3 dimensions (heptacube) as being its fundamental region and the cube in the center of the heptacube corresponds geometrically to the codeword. Also in \cite{perfcodes} it is conjectured that this is the only case for which a close-packing exists in dimension 3.  

Since the $7^{3}$ cubic lattice has a finite number of qubits, we use operations modulo 7 to guarantee that the qubits of a given coset (heptacube) remain inside the $7^{3}$ cubic lattice.

The $7^{3}$ cubic lattice consists of seven $7\times 7$ arrays which are $7\times 7$ cross-sections, where each square of these arrays means a cube. It is shown in \cite{ijam} that the seven codewords of each $7\times 7$ array can be generated by the vector $(0 \ 1 \ 4)$, that is, we can put them apart by one unit to the right in the horizontal direction and four units down in the vertical one. Now to rise in the third dimension to change from one array to the one above to continue the construction of the corresponding codewords we use the vector $(1 \ 0 \ 2)$ as the corresponding generator. Consequently, by using these vectors, we can construct the corresponding 49 codewords from the $7^{3}$ cubic lattice and the $7\times 7$ cross-sections can be labeled by the numbers 0, 1, 2, 3, 4, 5 and 6.

By knowing these generator vectors, we obtain the following result. 

\begin{theorem}\label{3LatticeCode}
The classical single-error-correcting code in the $7^{3}$ cubic lattice can be characterized as a lattice code. 
\end{theorem}  
\begin{proof} The authors in \cite{CibeleQIP} provide for $q=7$ the $2\times 2$ matrix $\left(
                 \begin{array}{cc}
                   2 & 1 \\
                   1 & 4 \\
                 \end{array}
               \right)$ whose line vectors generate the codewords of each $7\times 7$ array (cross-section) from the $7^{3}$ cubic lattice. Therefore, since the vector $(1 \ 0 \ 2)$ is another generator of the classical single-error-correcting code from the $7^{3}$ cubic lattice, we complete this matrix with the generator $(1 \ 0 \ 2)$ to obtain the $3\times 3$ matrix $A=\left(
                 \begin{array}{ccc}
                   0 & 2 & 1 \\
                   0 & 1 & 4 \\
                   1 & 0 & 2 \\
                 \end{array}
               \right)$ to the dimension 3 that, consequently, generates the sublattice $A \mathbb{Z}^{3}$ from $\mathbb{Z}^{3}$. Therefore, $\dfrac{V(A \mathbb{Z}^{3})}{V(\mathbb{Z}^{3})} = |\det A| = 7$, since $|\det A|=7$. 
               
As $V(\mathbb{Z}^{3})=1$, then $V(A \mathbb{Z}^{3})=7$, that is, the volume of the fundamental region of the sublattice $A \mathbb{Z}^{3}$ of the lattice $\mathbb{Z}^{3}$ is 7. From there, it is possible to observe that $\dfrac{V(7 \mathbb{Z}^{3})}{V(\mathbb{Z}^{3})}=7^{3}>7=\dfrac{V(A \mathbb{Z}^{3})}{V(\mathbb{Z}^{3})}$ and, thus, $V(7 \mathbb{Z}^{3})=7^{3}>7=V(A \mathbb{Z}^{3})$.

Hereupon it is possible to show that $A \mathbb{Z}^{3}\supset 7 \mathbb{Z}^{3}$, in fact: let $(7\alpha,7\beta,7\gamma)$, where $\alpha,\beta,\gamma \in \mathbb{Z}$, an arbitrary element of $7 \mathbb{Z}^{3}$. Then it is needed to show that there exists $v=(x,y,z)\in \mathbb{Z}^{3}$ such that $v\cdot A=(7\alpha,7\beta,7\gamma)$. Thenceforth, by straightforward computations, we obtain the solution $x=2\alpha+4\beta-\gamma$, $y=-4\alpha-\beta+2\gamma$ and $z=7\alpha$. Consequently, there exists  $v=(x,y,z)\in \mathbb{Z}^{3}$ such that $v\cdot A=(7\alpha,7\beta,7\gamma)$ and, therefore, $A \mathbb{Z}^{3}\supset 7 \mathbb{Z}^{3}$.  

Now since $V(7 \mathbb{Z}^{3})=7^{3}>7=V(A \mathbb{Z}^{3})$ and $A \mathbb{Z}^{3}\supset 7 \mathbb{Z}^{3}$, then we have the following nested lattice chain
\begin{equation}\label{chain1}
\mathbb{Z}^{3}\supset A \mathbb{Z}^{3}\supset 7 \mathbb{Z}^{3}.
\end{equation}

From \cite{forney} we obtain $\mid \mathbb{Z}^{3} / 7 \mathbb{Z}^{3} \mid =\dfrac{V(7 \mathbb{Z}^{3})}{V(\mathbb{Z}^{3})}=7^{3}$ and $\mid \mathbb{Z}^{3} / A \mathbb{Z}^{3} \mid=\dfrac{V(A \mathbb{Z}^{3})}{V(\mathbb{Z}^{3})}=7$, consequently, $\mid A \mathbb{Z}^{3} / 7 \mathbb{Z}^{3} \mid=\dfrac{V(7 \mathbb{Z}^{3})}{V(A \mathbb{Z}^{3})}=\dfrac{7^{3}}{7}=7^{2}=49$. Therefore the lattice quotient $A \mathbb{Z}^{3} / 7 \mathbb{Z}^{3}$ is the group consisting of the 49 cosets of $7 \mathbb{Z}^{3}$ in $A \mathbb{Z}^{3}$ and the set of these 49 cosets defines a lattice code. 

Now as the line vectors of the matrix $A$ which generates the sublattice $A \mathbb{Z}^{3}$ are the generators of the classical code from the $7^{3}$ cubic lattice and the lattice quotient $A \mathbb{Z}^{3} / 7 \mathbb{Z}^{3}$ provides operations modulo 7 over the respective vectors, then there exists a natural group  isomorphism between the classical code and the lattice code $A \mathbb{Z}^{3} / 7 \mathbb{Z}^{3}$, consequently, such a classical code is characterized by the lattice code $A \mathbb{Z}^{3} / 7 \mathbb{Z}^{3}$. 
\end{proof}

From this lattice code a new three-dimensional toric quantum code can be featured; in fact, the corresponding three-dimensional toric quantum code which is associated with the fundamental region as being the heptacube is established in the same way as it is in the three-dimensional toric quantum code for $q=7$ studied in \cite{Temperature}. Thenceforward, as in \cite{BombinDelgado}, the respective code length is decreased and it is given by the number of faces of the heptacube, that is, $n=3q=21$, since the heptacube has volume $q=7$ and each face belongs simultaneously to two cubes of the $q^{3}=7^{3}=343$ cubic lattice (lattice). The code dimension is $k=3$, since this three-dimensional toric quantum code is constructed on the $T^{3}$ torus which has genus $g=1$. From the fact that the $\mathbb{Z}^{3}$-lattice is self-dual, the code distance is defined as the minimum number of faces in the $q^{3}=7^{3}=343$ cubic lattice between two codewords. Such a distance is called \textit{Mannheim distance} and it is given by $d_{M} = min \{ \mid x \mid + \mid y \mid + \mid z \mid \ \ \mid \ (x,y,z) \in \mathcal{C} \}$, where $\mathcal{C}$ indicates the set of the 49 codewords and $w_{M} (x,y,z) = \mid x \mid + \mid y \mid + \mid z \mid$ is known as the \textit{Mannheim weight} of $(x,y,z)$. 

\begin{theorem}\label{3distance}
The minimum Mannheim distance of the corresponding three-dimensional toric quantum code is given by $3$.
\end{theorem}
\begin{proof}
From \cite{CibeleQIP}, the line vectors of the $2\times 2$ matrix $\left(
                 \begin{array}{cc}
                   2 & 1 \\
                   1 & 4 \\
                 \end{array}
               \right)$ generate the codewords of each $7\times 7$ array (cross-section) from the $7^{3}$ cubic lattice. On the other hand, in \cite{CibeleQIP}, Lemma 1, Lemma 2 and Example 1 provide the minimum Mannheim distance of the codewords of each $7\times 7$ array (cross-section) which is given by 3.  
               
As we have seen previously we use the vector $(1 \ 0 \ 2)$ which is the corresponding generator in the vertical direction to rise in the third dimension to change from one array to the one above to continue the construction of the corresponding codewords.

Since the minimum Mannheim distance of the codewords of each $7\times 7$ array (cross-section) is given by 3 and the Mannheim weight of $(1 \ 0 \ 2)$ is 3, then the minimum Mannheim distance of the corresponding three-dimensional toric quantum code is given by $3$.
\end{proof}

From there, the parameters of the new three-dimensional toric quantum code are $[[3q=21,3,3]]$.

Consequently, under the algebraic point of view, such a new three-dimensional toric quantum code can be characterized as the group consisting of the cosets of $7 \mathbb{Z}^{3}$ in $A \mathbb{Z}^{3}$. 

The duality of the lattice in which a toric quantum code is constructed influences the error correction pattern \cite{LivroBombin}, that is, if the corresponding lattice is self-dual, then the quantum channel without memory is symmetric; but if the corresponding lattice is not self-dual, then the quantum channel without memory is not symmetric. Therefore, as the $\mathbb{Z}^{3}$-lattice is self-dual, then the respective quantum channel without memory is symmetric.

\section{New Four-Dimensional Toric Quantum Code from a Lattice Code}\label{four}

Four-dimensional toric quantum codes are studied in \cite{4D}. Under the algebraic point of view, such toric quantum codes can be characterized as the group consisting of the cosets of $q \mathbb{Z}^{4}$ in $\mathbb{Z}^{4}$ \cite{Edson}, which in turn is isomorphic to $\mathbb{Z}_{q} \times \mathbb{Z}_{q} \times \mathbb{Z}_{q} \times \mathbb{Z}_{q} = \mathbb{Z}_{q}^{4}$ (4-dimensional hypercube), where $q$ is a positive integer. The identifications of the opposite hyperfaces of the region delimited by $\mathbb{Z}_{q}^{4}$ result in its identification with the four-dimensional torus denoted by $T^{4}$ which has genus $g=1$; for the sake of simplicity, we call this region \textit{lattice} or $q\times q \times q \times q = q^{4}$ \textit{hypercubic lattice}. The hypervolume associated with the lattice $\mathbb{Z}_{q}^{4}$ is $q^{4}$. 

By following the reasoning of the last section, there are $6q^{4}$ hyperfaces in the $q^{4}$ hypercubic lattice, that is, $n=6q^{4}$ qubits, where $n$ is the length of the code that is given by the number of hyperfaces of the $q^{4}$ hypercubic lattice. In the construction of these four-dimensional toric quantum codes \cite{4D}, such qubits are in a biunivocal correspondence with the hyperfaces of the $q^{4}$ hypercubes of the $q^{4}$ hypercubic lattice. Since $\mathbb{Z}^{4}$ is a self-dual lattice, then the code distance which is defined as being the minimum number of hyperfaces in the $q^{4}$ hypercubic lattice between two codewords is equal to $q^{2}$ \cite{4D,Edson}. Consequently, this class of codes has parameters $[[6q^{4},6,q^{2}]]$.

In this section, let $q=9$. The goal of this section is to construct a four-dimensional toric quantum code by using the fundamental region of a sublattice of the lattice $\mathbb{Z}^{4}$ that has hypervolume 9, since 9 divides the hypervolume $q^{4}=9^{4}$ of the lattice $\mathbb{Z}_{9}^{4}$. 

In \cite{perfcodes} the authors construct a classical single-error-correcting code in dimension four in the $9^{4}$ hypercubic lattice which has $729=9^{3}=q^{3}$ codewords. Each codeword of this classical code has the Lee sphere of radius 1 in 4 dimensions as being its fundamental region and the hypercube in the center of such Lee sphere corresponds geometrically to the codeword. Therefore each codeword can be featured as a central hypercube which has 8 hyperfaces to which another hypercube has been affixed to each of its hyperfaces. In \cite{perfcodes} it is also conjectured that this is the only case for which a close-packing exists in dimension 4.  

Since the $9^{4}$ hypercubic lattice has a finite number of qubits, we use operations modulo 9 to guarantee that the qubits of a given coset (Lee sphere of radius 1 in 4 dimensions (fundamental region)) remain inside the $9^{4}$ hypercubic lattice.

The $9^{4}$ hypercubic lattice consists of nine $9\times 9\times 9=9^{3}$ hypertorus which are $9\times 9\times 9=9^{3}$ cross-sections. Now each $9^{3}$ hypertorus consists of nine $9\times 9$ arrays, where each square of these arrays means a cube. It is shown in \cite{ijam} that the 9 codewords of each $9\times 9$ array can be generated by the vector $(0 \ 0 \ 1 \ 6)$, that is, we can put them apart by one unit to the right in the horizontal direction and six units down in the vertical one. Now to rise in the third dimension to change from one array to the one above to continue the construction of the 81 codewords of the corresponding hypertorus we use the vector $(0 \ 1 \ 1 \ 1)$ as the corresponding generator vector. Hence to generate all the 81 codewords of a cross-section (hypertorus) we must use the generator vectors $(0 \ 0 \ 1 \ 6)$ and $(0 \ 1 \ 1 \ 1)$. Finally to rise in the fourth dimension to change from one cross-section to the other above to continue the construction of the corresponding 729 codewords we use the vector $(1 \ 0 \ 0 \ 2)$ as the corresponding generator vector. Consequently, by using these vectors, we can construct the corresponding 729 codewords from the $9^{4}$ hypercubic lattice and the $9^{3}$ cross-sections can be labeled by the numbers 0, 1, 2, 3, 4, 5, 6, 7 and 8.

By knowing these generator vectors, we obtain the following result. 

\begin{theorem}\label{4LatticeCode}
The classical single-error-correcting code in the $9^{4}$ hypercubic lattice can be characterized as a lattice code. 
\end{theorem}  
\begin{proof}
By following the reasoning of Theorem \ref{3LatticeCode}, the authors in \cite{CibeleQIP} provide for $q=9$ the $2\times 2$ matrix $\left(
                 \begin{array}{cc}
                   1 & 6 \\
                   -1 & 3 \\
                 \end{array}
               \right)$ whose line vectors generate the 9 codewords of each $9\times 9$ array from the $9^{3}$ cross-sections (hypertorus). Therefore, since the vectors $(0 \ 1 \ 1 \ 1)$ and  $(1 \ 0 \ 0 \ 2)$ are other two generator vectors of the classical single-error-correcting code from the $9^{4}$ hypercubic lattice, we complete this matrix with these generators to obtain the $4\times 4$ matrix $A=\left(
                 \begin{array}{cccc}
                   0 & 0 & 1 & 6 \\
                   0 & 0 & -1 & 3 \\
                   0 & 1 & 1 & 1 \\
                   1 & 0 & 0 & 2 \\
                 \end{array}
               \right)$ to the dimension 4 that, consequently, generates the sublattice $A \mathbb{Z}^{4}$ from $\mathbb{Z}^{4}$. Therefore, since $ \mid \det A \mid =9$, $\dfrac{V(A \mathbb{Z}^{4})}{V(\mathbb{Z}^{4})} = \mid \det A \mid = 9$. 
               
As $V(\mathbb{Z}^{4})=1$, then $V(A \mathbb{Z}^{4})=9$, that is, the volume of the fundamental region of the sublattice $A \mathbb{Z}^{4}$ of the lattice $\mathbb{Z}^{4}$ is 9. From there, it is possible to observe that $\dfrac{V(9 \mathbb{Z}^{4})}{V(\mathbb{Z}^{4})}=9^{4}>9=\dfrac{V(A \mathbb{Z}^{4})}{V(\mathbb{Z}^{4})}$ and, thus, $V(9 \mathbb{Z}^{4})=9^{4}>9=V(A \mathbb{Z}^{4})$.

Hereupon it is possible to show that $A \mathbb{Z}^{4}\supset 9 \mathbb{Z}^{4}$, in fact: let $(9\alpha,9\beta,9\gamma,9\delta)$, where $\alpha,\beta,\gamma,\delta \in \mathbb{Z}$, an arbitrary element of $9 \mathbb{Z}^{4}$. Then it is needed to show that there exists $v=(x,y,z,w)\in \mathbb{Z}^{4}$ such that $v\cdot A=(9\alpha,9\beta,9\gamma,9\delta)$. Thenceforth, by straightforward computations, we obtain the solution $x=2\alpha-4\beta+3\gamma+\delta$, $y=2\alpha+5\beta-6\gamma+\delta$, $z=9\beta$ and $w=9\alpha$. Consequently, there exists  $v=(x,y,z,w)\in \mathbb{Z}^{4}$ such that $v\cdot A=(9\alpha,9\beta,9\gamma,9\delta)$ and, therefore, $A \mathbb{Z}^{4}\supset 9 \mathbb{Z}^{4}$.  

Now since $V(9 \mathbb{Z}^{4})=9^{4}>9=V(A \mathbb{Z}^{4})$ and $A \mathbb{Z}^{4}\supset 9 \mathbb{Z}^{4}$, then we have the following nested lattice chain
\begin{equation}\label{chain2}
\mathbb{Z}^{4}\supset A \mathbb{Z}^{4}\supset 9 \mathbb{Z}^{4}.
\end{equation}

From \cite{forney} we obtain $\mid \mathbb{Z}^{4} / 9 \mathbb{Z}^{4} \mid =\dfrac{V(9 \mathbb{Z}^{4})}{V(\mathbb{Z}^{4})}=9^{4}$ and $\mid \mathbb{Z}^{4} / A \mathbb{Z}^{4} \mid=\dfrac{V(A \mathbb{Z}^{4})}{V(\mathbb{Z}^{4})}=9$, consequently, $\mid A \mathbb{Z}^{4} / 9 \mathbb{Z}^{4} \mid=\dfrac{V(9 \mathbb{Z}^{4})}{V(A \mathbb{Z}^{4})}=\dfrac{9^{4}}{9}=9^{3}=729$. Therefore the lattice quotient $A \mathbb{Z}^{4} / 9 \mathbb{Z}^{4}$ is the group consisting of the 729 cosets of $9 \mathbb{Z}^{4}$ in $A \mathbb{Z}^{4}$ and the set of these 729 cosets defines a lattice code. 

Now as the line vectors of the matrix $A$ which generates the sublattice $A \mathbb{Z}^{4}$ are the generators of the classical code from the $9^{4}$ hypercubic lattice and the lattice quotient $A \mathbb{Z}^{4} / 9 \mathbb{Z}^{4}$ provides operations modulo 9 over the respective vectors, then there exists a natural group isomorphism between the classical code and the lattice code $A \mathbb{Z}^{4} / 9 \mathbb{Z}^{4}$, consequently, such a classical code is characterized by the lattice code $A \mathbb{Z}^{4} / 9 \mathbb{Z}^{4}$. 
\end{proof}

From this lattice code a new four-dimensional toric quantum code can be featured; in fact, the corresponding four-dimensional toric quantum code which is associated with the fundamental region as being the Lee sphere of radius 1 in 4 dimensions is established in the same way as it is in the four-dimensional toric quantum code for $q=9$ studied in \cite{4D}. Thenceforward, as in \cite{BombinDelgado}, the respective code length is decreased and it is given by the number of hyperfaces of the corresponding Lee sphere of radius 1, that is, $n=6q=54$, since such a Lee sphere has hypervolume $q=9$ and each hyperface belongs simultaneously to two hypercubes of the $q^{4}=9^{4}=6561$ hypercubic lattice (lattice). The code dimension is $k=6$, since this four-dimensional toric quantum code is constructed on the $T^{4}$ torus which has genus $g=1$. From the fact that the $\mathbb{Z}^{4}$-lattice is self-dual, the code distance is defined as the minimum number of hyperfaces in the $q^{4}=9^{4}=6561$ hypercubic lattice between two codewords. Such a distance is called \textit{Mannheim distance} and it is given by $d_{M} = min \{ \mid x \mid + \mid y \mid + \mid z \mid + \mid w \mid \ \ \mid \ (x,y,z,w) \in \mathcal{C} \}$, where $\mathcal{C}$ indicates the set of the 729 codewords and $w_{M} (x,y,z,w) = \mid x \mid + \mid y \mid + \mid z \mid + \mid w \mid$ is known as the \textit{Mannheim weight} of $(x,y,z,w)$. 

\begin{theorem}\label{4distance}
The minimum Mannheim distance of the corresponding four-dimensional toric quantum code is given by $3$.
\end{theorem}
\begin{proof}
By following the reasoning of Theorem \ref{3distance}, from \cite{CibeleQIP}, the line vectors of the $2\times 2$ matrix $\left(
                 \begin{array}{cc}
                   1 & 6 \\
                   -1 & 3 \\
                 \end{array}
               \right)$ generate the 9 codewords of each $9\times 9$ array from the $9^{3}$ cross-sections (hypertorus). On the other hand, in \cite{CibeleQIP}, Lemma 1, Lemma 2 and Example 2 provide the minimum Mannheim distance of the codewords of each $9\times 9$ array which is given by 3.  
               
As we have seen previously we use the vectors $(0 \ 1 \ 1 \ 1)$ and $(1 \ 0 \ 0 \ 2)$ which are the corresponding generator in the vertical directions to rise in the third and four dimensions, respectively, to change from one array to the one above and from one cross-section to the other above, respectively, to continue the construction of the corresponding codewords.         

Since the minimum Mannheim distance of the codewords of each $9\times 9$ array is given by 3 and the Mannheim weight of $(0 \ 1 \ 1 \ 1)$ and $(1 \ 0 \ 0 \ 2)$ is 3, then the minimum Mannheim distance of the corresponding four-dimensional toric quantum code is given by $3$.
\end{proof}

From there, the parameters of the new four-dimensional toric quantum code are $[[6q=54,6,3]]$.

Consequently, under the algebraic point of view, such a new four-dimensional toric quantum code can be characterized as the group consisting of the cosets of $9 \mathbb{Z}^{4}$ in $A \mathbb{Z}^{4}$. 

As it was noted in Section \ref{three}, the duality of the lattice in which a toric quantum code is constructed influences the error correction pattern \cite{LivroBombin}, therefore, as the $\mathbb{Z}^{4}$-lattice is self-dual, then the respective quantum channel without memory is symmetric.

\section{New Three and Four-Dimensional Quantum Burst-Error-Correcting Codes}

The authors in \cite{CibeleQIP} provide a new class of toric quantum codes by constructing a classical cyclic code on the square lattice $\mathbb{Z}_{q}\times \mathbb{Z}_{q}$ for all odd integers $q\geq 5$ and, consequently, new toric quantum codes are constructed on such square lattices via a combinatorial method, that is, regardless of whether $q$ can be represented as a sum of two squares. Moreover, they propose a quantum interleaving technique by using these constructed two-dimensional toric quantum codes which shows that the code rate and the coding gain of the interleaved two-dimensional toric quantum codes are better than the code rate and the coding gain of Kitaev toric quantum codes for $q = 2n + 1$, where $n\geq 2$, and of an infinite class of Bombin and Martin-Delgado toric quantum codes.

The two-dimensional quantum interleaving technique proposed in \cite{CibeleQIP} is extended in this section to the dimensions three and four by applying it to the three and four-dimensional toric quantum codes previously obtained in Sections \ref{three} and \ref{four}, respectively. As a consequence, new three and four-dimensional quantum burst-error correcting codes are obtained which have better information rates and in the end of this section the appropriate comparisons are discussed and presented in tables. The following Theorem provides the parameters of these new three and four-dimensional quantum burst-error correcting codes.

\begin{theorem}
Consider $n_{D}=3,4$ and $q=7,9$, respectively, where $n_{D}$ is the dimension of the $q^{n_{D}}$ hypercubic lattice. The combination of the three and four-dimensional toric quantum codes provided in Sections \ref{three} and \ref{four}, respectively, and the interleaving technique results respectively in three and four-dimensional interleaved toric quantum codes with parameters $[[3q^{3},3q^{2},t_{i}=q]]$ ($q=7$) and $[[6q^{4},6q^{3},t_{i}=q]]$ ($q=9$), respectively, where $t_{i}$ is the interleaved toric quantum code error correcting capability.
\end{theorem}
\begin{proof}
The parameters of the three and four-dimensional toric quantum codes presented in Sections \ref{three} and \ref{four}, respectively, are given respectively by $[[3q=21,k=3,d_{M}=3]]$ ($q=7$) and $[[6q=54,k=6,d_{M}=3]]$ ($q=9$). A toric quantum code with minimum distance $d_{M}$ is able to correct up to $t$ errors, where $t=\lfloor \frac{d_{M}-1}{2} \rfloor$ \cite{ClariceArtigo}, therefore, such codes are able to correct up to $t=1$ error. 

We assume that the clusters of errors have the shape of the Lee sphere of radius 1 in $n_{D}$ dimensions and the qubits are in a biunivocal correspondence with the faces (hyperfaces) of the $q^{n_{D}}$ hypercubic lattice. Figure \ref{fig9} shows the model of an storage system under consideration.
\begin{figure}[ht]%
\centering
\includegraphics[width=0.7\textwidth]{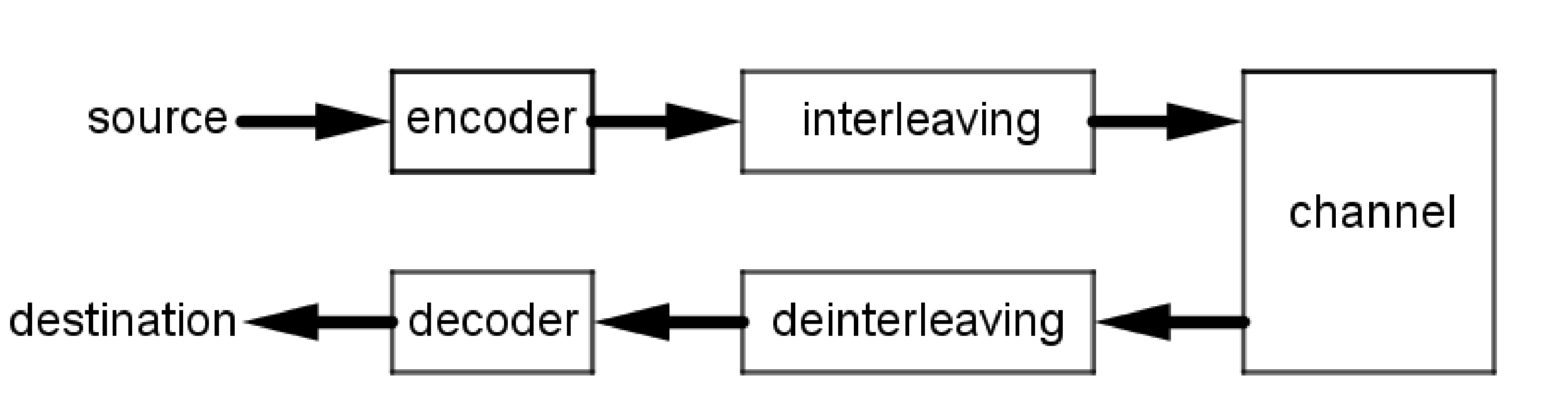}
\caption{Model of an storage system.}
\label{fig9}
\end{figure}

Consider the $q^{n_{D}}$ hypercubic lattice, where $n_{D}=3,4$ and $q=7,9$, respectively. Such a hypercubic  lattice has a total of $\alpha q^{n_{D}}$ qubits (hyperfaces), where $\alpha=3$ for $q=7$ and $n_{D}=3$, and $\alpha=6$ for $q=9$ and $n_{D}=4$. From now on we explain how to spread such $\alpha q^{n_{D}}$ adjacent qubits in the $q^{n_{D}}$ hypercubic lattice. Observe that the three and four-dimensional toric quantum codes presented in Sections \ref{three} and \ref{four}, respectively, in the respective $q^{n_{D}}$ hypercubic lattice consist of $q^{(n_{D}-1)}$ codewords whose fundamental region is the Lee sphere of radius 1 in $n_{D}$ dimensions which has $q$ hypercubes and $\alpha q$ qubits (hyperfaces).   

Thereby the first $\alpha q^{(n_{D}-1)}$ adjacent qubits which are related to the cross-section $j=0$ are spread over the hypercubes that correspond to the $q^{(n_{D}-1)}$ codewords of the $q^{n_{D}}$ hypercubic lattice as it follows: firstly, ordering such adjacency consists of using the order of the codewords of the cross-section $j=0$ where such an order is featured by the generator vectors that are used in the generation of the codewords of each cross-section $j$ ($j=0,1,\ldots,q-1$); observe that these generator vectors are the same for generating the corresponding codewords of each cross-section $j$. Thenceforward, the first $q^{(n_{D}-1)}$ qubits (hyperfaces) are placed on the hypercubes that correspond to the $q^{(n_{D}-1)}$ codewords of the $q^{n_{D}}$ hypercubic lattice; notice that it is needed to follow the sequence of the codewords of the respective code constructed over the $q^{n_{D}}$ hypercubic lattice by using all the generator lattice vectors and initiating with the null lattice vector (null codeword) which belongs to the cross-section $j=0$. The other $\alpha-1$ blocks of $q^{(n_{D}-1)}$ qubits are placed in the same way on the same hypercubes (related to the $q^{(n_{D}-1)}$ codewords of the code).   

Each cross-section $j$ ($j=0,1,\ldots,q-1$) has $q^{(n_{D}-2)}$ codewords and each codeword has $\alpha q$ qubits, then each cross-section $j$ has $\alpha q^{(n_{D}-1)}$ qubits. Now the other $q-1$ blocks with $\alpha q^{(n_{D}-1)}$ adjacent qubits which correspond to the cross-sections $j=1,2,\ldots,q-1$, respectively, are spread across the $q^{n_{D}}$ hypercubic lattice analogously, that is, suppose that we wish to spread the next block ($j=1$) with $\alpha q^{(n_{D}-1)}$ adjacent qubits (hyperfaces) across the $q^{n_{D}}$ hypercubic lattice. As the corresponding fundamental region is the Lee sphere of radius 1 in $n_{D}$ dimensions which has $q$ hypercubes by including the one related to the codeword, then we start the corresponding arrangement by placing the first qubit of the $q^{(n_{D}-1)}$ first qubits in one of the $q-1$ hypercubes related to the Lee sphere of radius 1 of the null codeword (the hypercube related to the null codeword was used to spread the qubits of the cross-section $j=0$) and, then, we use all the generator lattice vectors to follow the sequence of the respective codewords to spread the other $q^{(n_{D}-1)}-1$ qubits over the corresponding hypercubes. The other $\alpha-1$ blocks of $q^{(n_{D}-1)}$ qubits are placed in the same way on the same hypercubes. 

Consequently, the $\alpha q^{(n_{D}-1)}$ qubits of each cross-section $j=2,\ldots,q-1$ are spread analogously and correspond to a hypercube of the $q-2$ remaining hypercubes of the Lee sphere of radius 1. Therefore such an interleaving technique shows that the qubits of a certain hypercube of the $q^{(n_{D}-1)}$ codewords of the $q^{n_{D}}$ hypercubic lattice are the $\alpha q^{(n_{D}-1)}$ qubits of the codewords of a cross-section $j$ ($j=0,1,\ldots,q-1$). 

Since the three and four-dimensional toric quantum codes presented in Sections \ref{three} and \ref{four}, respectively, are able to correct up to $t=1$ error, then by using the interleaving method previously described it is possible to correct up to $q$ errors in burst from the total of $\alpha q^{(n_{D})}$ qubits. In fact, assume that a cluster of $q$ errors in the quantum channel has the shape of the corresponding fundamental region which is the Lee sphere of radius 1. Therefore, when the deinterleaving process is applied, each one of these $q$ errors occurs, respectively, in a codeword of a different cross-section $j$ ($j=0,1,\ldots,q-1$) and then the respective three and four-dimensional toric quantum codes are applied to correct these $q$ errors in burst. 

Therefore, in this section, it is shown that the combination of the three and four-dimensional toric quantum codes provided in Sections \ref{three} and \ref{four}, respectively, whose parameters are  $[[3q=21,k=3,d_{M}=3]]$ ($q=7$) and $[[6q=54,k=6,d_{M}=3]]$ ($q=9$), respectively, and the presented interleaving method results, respectively, in three and four-dimensional interleaved toric quantum codes whose parameters are given respectively by $[[3q^{3},3q^{2},t_{i}=q]]$ ($q=7$) and $[[6q^{4},6q^{3},t_{i}=q]]$ ($q=9$), where $t_{i}$ is the interleaved toric quantum code error correcting capability. As a consequence, new three and four-dimensional quantum burst-error correcting codes are obtained by applying this interleaving method to the respective three and four-dimensional toric quantum codes. \end{proof}

The code rate \cite{rates} and the coding gain \cite{rates} are given by $R=\dfrac{k}{n}$ and $G=\dfrac{k}{n}(t+1)$, respectively, where $t$ is the toric quantum code error correcting capability. 

Table \ref{T1} shows that the code rate and the coding gain (in dB) of the three and four-dimensional toric quantum codes provided in Sections \ref{three} and \ref{four}, respectively, are better than the code rate and the coding gain of the three and four-dimensional toric quantum codes from the literature whose parameters are $[[3q^{3},3,t=3]]$ ($q=7$) \cite{Temperature} and $[[6q^{4},6,t=40]]$ ($q=9$) \cite{4D}, respectively. 

\begin{table}[h!]
\begin{center}
\caption{Code rate and coding gain of the three and four-dimensional toric quantum codes from the literature \cite{Temperature,4D} and Sections \ref{three} and \ref{four}, respectively}
\label{T1}
\begin{tabular}{| c | c | c |}
\hline
3D and 4D & Code Rate & Coding Gain \\
Toric Quantum Codes & $R=\dfrac{k}{n}$ & $G=\dfrac{k}{n}(t+1)$ (dB) \\
\hline
$[[3q^{3},3,t=3]]$ ($q=7$) & 0,00292 & 0,01168 \\
\hline
$[[6q^{4},6,t=40]]$ ($q=9$) & 0,000152 & 0,006232 \\ 
\hline
$[[3q=21,k=3,t=1]]$ ($q=7$) & 0,1429 & 0,2858 \\
\hline
$[[6q=54,k=6,t=1]]$ ($q=9$) & 0,1111 & 0,2222 \\
\hline 
\end{tabular}
\end{center}
\end{table}

Table \ref{T2} shows the equivalent interleaved three and four-dimensional toric quantum codes (three and four-dimensional quantum burst-error-correcting codes) and their corresponding interleaving code rates $R_{i}$ and coding gains $G_{i}$ (in dB) when the cluster of errors has the shape of the corresponding Lee sphere of radius 1 over the $q^{n_{D}}$ hypercubic lattice. Notice that the coding gain of the three and four-dimensional quantum burst-error correcting codes is better than the coding gain of the three and four-dimensional toric quantum codes provided in Sections \ref{three} and \ref{four}, respectively, and the code rate of the three and four-dimensional quantum burst-error correcting codes is equal to the code rate of the three and four-dimensional toric quantum codes provided in Sections \ref{three} and \ref{four}, respectively. 

Moreover, the code rate and the coding gain (in dB) of the three and four-dimensional quantum burst-error correcting codes are better than the code rate and the coding gain of the three and four-dimensional toric quantum codes from the literature whose parameters are $[[3q^{3},3,t=3]]$ ($q=7$) \cite{Temperature} and $[[6q^{4},6,t=40]]$ ($q=9$) \cite{4D}, respectively.   

\begin{table}[h!]
\begin{center}
\caption{Code rate and coding gain of the three and four-dimensional quantum burst-error-correcting codes from the interleaving method}
\label{T2}
\begin{tabular}{| c | c | c |}
\hline
Interleaved 3D and 4D & Code Rate & Coding Gain \\
Toric Quantum Codes & $R_{i}=\dfrac{k}{n}$ & $G_{i}=\dfrac{k}{n}(t_{i}+1)$ (dB) \\
\hline
$[[3q^{3},3q^{2},t_{i}=q]]$ ($q=7$) & 0,1429 & 1,1432 \\
\hline
$[[6q^{4},6q^{3},t_{i}=q]]$ ($q=9$) & 0,1111 & 1,111 \\ 
\hline 
\end{tabular}
\end{center}
\end{table}

\section{Conclusion}

References \cite{ijam,celso} provide effective interleaving techniques for combating burst of errors by using classical codes. However, to the best of our knowledge, little is known regarding interleaving techniques for combating cluster of errors in toric quantum codes. While in classical interleaving \cite{ijam,celso} each face represents a bit and, therefore, the bits sequentially allocated to each face are shifted according to the generator of the group, and so on, in the quantum interleaving used in this work the use of this sight to qubits which are allocated in each fundamental region is presented consistently. In addition, in the classical interleaving techniques \cite{ijam,celso} several classical error-correcting codes are needed to correct the corresponding errors in burst; now, in this work, we use a quantum interleaving method which provides a single quantum burst-error-correcting code in dimensions three and four that has better information rates than those three and four-dimensional toric quantum codes from the literature.

\section{Acknowledgments}\nonumber
\noindent The authors would like to thank the financial Brazilian agency CNPq (Conselho Nacional de Desenvolvimento Cient\'{i}fico e Tecnol\'{o}gico - Brazil) for the funding support and under grant no. 101862/2022-9 (chamada CNPq 25/2021, PDS 2021).

\end{document}